\documentclass[article,longbibliography,
 amsmath,amssymb,
 prl,superscriptaddress]{revtex4-2}
\usepackage{graphicx}
\usepackage{dcolumn}
\usepackage{bm}
\usepackage[colorlinks=true, citecolor=blue, linkcolor=blue, urlcolor=blue]{hyperref}
\usepackage{xcolor}
\usepackage{amsmath, amssymb, amsfonts,amsthm}
\theoremstyle{definition}

\newtheorem{theorem}{Theorem}

\newtheorem{lemma}{Lemma}
\newtheorem{corollary}{Corollary}

\usepackage{mathtools}
\usepackage{subfigure}
\usepackage{mathrsfs}
\usepackage{sidecap}
\usepackage{verbatim}

\definecolor{purple1}{rgb}{128,0,128}

\usepackage{bigints}
\newcommand{\bea}{\begin{eqnarray}}
\newcommand{\ea}{\end{eqnarray}}

\definecolor{darkpastelgreen}{rgb}{0.01, 0.75, 0.24}
\begin{document}
\title{A Revisit to Lorentz Transformation without Light}
\author{Satadal Datta} 
\affiliation{%
Seoul National University, Department of Physics and Astronomy, Center for Theoretical Physics, Seoul 08826, Korea}
\date{\today}
\begin{abstract}
We derive rotation free Lorentz Transformation (LT) between two inertial reference frames without using the second postulate of Einstein, i.e., we do not assume the invariant speed of light (in vacuum) under LT. Using principle of relativity, homogeneity of space and time, and isotropy of space, we find a general transformation rule between two inertial frames where a speed, invariant under that transformation, arises naturally. This idea first came into light by Mathematician Ignatowski \cite{von1911relativitatsprinzip} around the year, 1910. Without loss of novelty, here we present our derivation while reviewing more recent literature in this subject. 

\end{abstract}
\maketitle
\section{Introduction} 
Einstein's 1905 work \cite{einstein} (for English translation, see \cite{saha1920principle}) changed the fundamental notions in Physics forever. It all started from the discrepancy between Maxwell's theory on Electromagnetic field and Newtonian Mechanics. There were attempts to save both of these ideas, i.e., by introducing the concept of ether medium as a preferred inertial frame where Maxwell's equations were proposed to be true. After Michelson-Morley's experiment nullified the existence of ether medium, then there were attempts to rescue Ether theory, i.e., the idea of Lorentz-FitzGerald contraction of interferometer arm along the direction of ether wind. Even `Lorentz-FitzGerald contraction' hypothesis was nullified by further experimental results (for Historical review Ref. \cite{resnick1992basic,bohm2015special}). Einstein came up with the novel idea beginning with two postulates. The first one is the principle of relativity, and the second one is constancy of speed of light (in vacuum) with respect to all inertial reference frames, regardless the motion of source; and he derived Lorentz Transformation (LT) equations, more precisely Lorentz Boost (LB) equations. To avoid confusion, whenever we mention LT hereafter, it means LB, there is no spatial rotation involved. It turns out that Maxwell's electromagnetic field equations hold in every inertial frame, but Newton's idea of absolute space and time is not true.  Time dilation of moving clock, length contraction of moving object along the direction of motion, Doppler effect and aberration of light etc \cite{einstein, resnick1992basic, griffiths2005introduction, rosser2017introductory, rindler1991introduction} are direct consequences of LT.

The idea of deriving LT equations only from the Einstein's first postulate is not new. To the best of our knowledge, Ignatowski first came up with this idea \cite{von1911relativitatsprinzip}. Later on, some Physicists came up with their own versions of rederiving Lorentz equation and reviewed Ignatowski's work until recently \cite{gorini,lee,sen,peli,ga,BC,pittphilsci13220} (see for a review on Ignatowski's idea, the Ref. \cite{torretti1996relativity}).

Here, without using the second postulate, we derive Lorentz transformation  from scratch. We begin from the known definitions, we point out the key aspects of Einstein's idea, and its differences from the derivation what we present here. We provide a brief overview on this topic to the readers by mentioning several works which were accomplished by many researchers with regard to Ignatowski's work. Our aim of this paper is to present the derivation to the readers as clear as possible, and with our new additions and insights, we have tried to make the topic simple and interesting. We have also tried to provide an overall understanding by referring to some exciting literature on this subject.

\section{A Brief Review of Definitions and Axioms}
{\it Inertial Reference Frame}, is defined as a reference frame (coordinate system) where Newton's 1st law of motion holds. Newton's 1st law provides us the definition of an inertial frame (IF). Since, we have not yet derived the LT equation, it would be imprecise to assume the validity of Newton's other two laws of motion (see for more discussion on this in the Ref. \cite{griffiths2005introduction}). In Newton's book \cite{newton1833philosophiae}, his idea of absolute space and time is consistent with his three laws of motion. Here we, as Einstein did, do not assume absoluteness of space and time; we start from the simplest set of axioms as mentioned below.

{\it Principle of Relativity} states that if we have an IF and another reference frame which moves with uniform velocity with respect to it, then laws of Nature take same form in both reference frames. Relativity principle asserts that if an inertial frame is given, any other reference frame that moves with uniform velocity with respect to it is itself inertial \cite{fock}.

{\it Homogeneity and Isotropy:-} We also assume homogeneity property of space and time, i.e., space and time intervals transform (from an IF to an another IF) in the same way at all space-time coordinates.
We also assume that there is no preferred spatial direction, i.e., isotropy property of space. 

\noindent
If Nature behaves {\it reasonably} as it appears so far!, we do believe that these assumptions are the only ways of abstraction of the reality, there is no other way. However, without relativity principle one can still define IF by only considering Newton's 1st law of motion to hold in all IFs. In this gedanken scenario, one can also consider linear coordinate transformations between IFs, we would like to recommend the Ref. \cite{VV} to the readers in this regard.

\noindent
For simplicity, we choose Cartesian coordinate system to work with in an IF, $F$. Three space and one time coordinate of a point in $F$ are $(x,y,z)$, and $t$. For a point object, $x=x(t),~y=y(t),~z=z(t)$, as time $t$ {\it flows forward}.


{\it Boost:-} we define boost as a bijective map from an IF to a reference frame moving with uniform velocity relative to it. We observe that bijective property of Cartesian coordinate transformation is {\it very} important because without it, a single Cartesian coordinate in one reference frame may have multiple images in the other.

{\it Length and time intervals:-} The rest length of a rigid rod stationary in $F$ is the Euclidean distance between its two edges which can be measured by a rigid ruler also at rest in $F$. If the rod moves with uniform velocity relative to $F$, the rest length measured in the rod's comoving frame $F'$ will not differ from its rest value in $F$ because of principle of relativity. Time, labelling succession of events, is measured by clock. By time in $F$, we mean time reading of a clock which is stationary at $F$. Similar to the length measurement, rest time measurement is independent of the clock's uniform velocity with respect to an IF.
The length of a moving rod in $F$ is the Euclidean distance between the $(x,y,z)$ coordinates of its two edges while the coordinates of the two edges are recorded simultaneously in $F$. To make such a simultaneous measurement, we need the stationary observers in $F$ to use the same label as succession of events in time measurement. We assume that irrespective of the stationary observers' positions in space, they agree on time measured by their identical {\it already} synchronized clocks. Basically time is just a coordinate $t$, independent of space coordinates $(x,y,z)$, just like $x$, $y$, $z$ are independent from each other. In the similar fashion, in the reference frame $F'$ moving with uniform velocity relative to $F$, the corresponding spacetime coordinate is $(x',y',z',t')$. However, there is a way to {\it practically} synchronize clocks at different locations in an IF as Einstein described in his work \cite{einstein}. In the next section, we first briefly review Einstein's idea. 

\section{Einstein's view} Einstein did a thought experiment to define synchronism of clocks at different locations in space. There are two observers at two different locations $A$ and $B$ in $F$, each of them has a clock at their respective locations. As a matter of simplicity, whenever we mention observer/clock/location in a reference frame, it means that observer/clock/location is stationary in that reference frame. The reading of the clock of the first observer at location $A$ is $t_{\rm A}$ just when a light signal from $A$ travels towards location $B$; the light signal is reflected from a stationary mirror at $B$ at $t=t_{\rm B}$ of the second observer's clock, the reflected light signal travels towards $A$, observer at $A$ receives the light signal at $t=t_{\rm A '}$ according to their clock. The  clocks at $A$ and $B$ are identical and synchronized if and only if
\begin{equation}\label{sy}
t_{\rm B}-t_{\rm A}=t_{\rm A'}-t_{\rm B}.
\end{equation}
The equality comes from the isotropy property of space; light signal propagates with same speed in all direction in space. In our case, we can still use light signal to synchronize clocks in an IF as described in above because in this thought experiment involving one IF, Einstein's second postulate is not used. If we {\it strictly} do not even want to use light, as some readers may prefer, one may imagine elastic collision between an object and a wall in place of light signal and the mirror; we leave this up to the reader.
Evidently, synchronism has transitive property, i.e., if a clock at location $A$ in $F$ and a clock at location $B$ in $F$ are synchronized with each other, and if the clock at $B$ is synchronized with another clock at location $C$ in $F$, then the clock at location $A$ and the clock at location $C$ are also synchronized with each other. This is how we establish an agreement on time readings of all stationary observers situated at all locations in $F$. 
If the relation \eqref{sy} is satisfied, the distance between location $A$ and $B$ is 
\begin{equation}
l_{AB}=c(t_{\rm B}-t_{\rm A})=\frac{c}{2}(t_{\rm A'}-t_{\rm A}),
\end{equation}
where $c$ is the speed of light in vacuum. The thought experiment is assumed to be performed in vacuum.
The time readings in $F'$ of the successive events in the aforementioned experiment performed in $F$ are $t'=t'_{A}$, $t'=t'_{B}$, and $t'=t'_{A'}$ respectively where $F'$ has a nonzero uniform velocity relative to $F$, and also let's assume the equality of Eq. \eqref{sy} holds, i.e., clocks at the locations $A$, and $B$ in $F$ are synchronized. Einstein found  \cite{einstein,saha1920principle} that the observers in $F'$ disagree on the simultaneity of the clocks in $F$ because the light signal used to check simultaneity in $F$ also travels with the same speed in $F'$ (the second postulate in action), 
\begin{equation}\label{sy2}
t'_{\rm B}-t'_{\rm A}\neq t'_{\rm A'}-t'_{\rm B}.
\end{equation}
Simply putting, observers in $F$ and the observers in the moving frame $F'$, label the succession of events differently. This is briefly the starting key idea of Einstein's derivation of LT. 

\section{Derivation without light}
Let's imagine a reference frame $F'$, equipped with three mutually perpendicular axes $X'$, $Y'$, $Z'$, is moving with a uniform velocity $v$ along the $X$ axis of $F$ (from now on we drop the word `IF' before $F$ and $F'$ for brevity). At $t=0$, the origin $O$ of $F$ coincides with the origin $O'$ of $F'$ at $t'=0$, i.e., $(x=0,y=0,z=0,t=0)$ in $F$ is $(x'=0,y'=0,z'=0,t'=0)$ in $F'$. Since boost is a bijective function, for any given $(x',y',z')$ in $F'$, there is a unique $(x,y,z)$ at each instant of time $t$ in $F$; and vice versa. Therefore, an observer in $F$ concludes that in their frame the moving origin $O'$ of $F'$ {\it uniquely} coincides with their origin $O$ at $t=0$, and observer in $F'$ concludes that in their frame the moving origin $O$ of $F$ {\it uniquely} coincides with their origin $O'$ at $t'=0$.

\noindent
Just like the origins of the two coordinate systems, we assume that another point $X'^*$ on the $X'$ axis of $F'$ coincides with a point $X^*$ on the $X$ axis of $F$ at $t=0$ and vice versa due to bijectivity. Since every stationary point in $F'$ is moving with uniform velocity in the positive $X$ direction of $F$, therefore coincidence of two points of two straight lines results in coincidence of those two straight lines as a whole (from Euclidean geometry).  In the same logic we can make the other two mutually perpendicular axes $Y'$, $Z'$ perpendicular to $X'$ axis of $F'$ to coincide with the $Y$, $Z$ axes respectively at $t=0$ in $F$.  At any $t>0$ in $F$, every stationary point in $F'$ has moved distance $vt$ along the $X$ axis, thus observers in $F$ {\it see} (more appropriate word instead of `see' is `measure' since we do not have light; we take these two words ambiguously for simplicity) boost as a uniform translation motion. Therefore, $X$ axis of $F$ and the $X'$ axis of $F'$ coincide with each other at any $t$ of $F$ and at any $t'$ of $F'$; the $Y$, $Y'$ axes remain parallel to each other at any $t$ of $F$ and at any $t'$ of $F'$; and the statement as for $Y,~Y'$ axes holds for $Z$, $Z'$ axes as well due to symmetry.

As a result, $X-Z$ plane of $F$ and $X'-Z'$ plane of $F'$ coincide with each other at any $t$ of $F$ and at any $t'$ of $F'$, and the same statement holds for the $X-Y$ plane of $F$ and $X'-Y'$ plane of $F'$; and on the other hand $Y-Z$ plane of $F$ and $Y'-Z'$ plane of $F'$ remain parallel to each other at any $t$ of $F$ and at any $t'$ of $F'$.

Further on, due to homogeneity of space and time, the bijective function which describes boost has to be a linear function. From the above discussion, the most general form of transformation from $F$ to $F'$ which one can imagine is:
\begin{eqnarray}
& x'=\gamma (v)(x-vt) \label{x1}\\
& y'=\phi (v) y, \label{y1}\\
& z'=\phi (v) z, \label{z1}\\
& t'=a(v)t +b(v)x+c(v)y+d(v)z, \label{t1}
\end{eqnarray}
where $\gamma(v)$, $\phi(v)$, $a(v)$, $b(v)$, $c(v)$, $d(v)$ are real functions of $v$.
\begin{lemma}\label{lemma 1}
$\gamma(v=0)=\phi(v=0)=a(v=0)=1$, $b(v=0)=c(v=0)=d(v=0)=0$, and $\gamma (v)>0$, $\phi(v)>0$.
\end{lemma}
\begin{proof}
Evidently $v=0$ has to correspond to the identity transformation. 
If at some $v$, $\gamma (v)$ is negative, then that implies that the relation \eqref{x1} makes $Y-Z$ plane of $F$ at $t=0$ very special to the observers in $F'$ because they perceive $Y-Z$ plane of $F$ as a mirror whereas all planes parallel to $Y-Z$ plane in $F$ must have equal preference at any $t$, at any $v$. Hence we have $\gamma (v)>0$. In the same spirit, $\phi (v)>0$.
\end{proof}
\begin{lemma}\label{lemma 2}
The assumption of homogeneity is compatible with relativity principle. 
\end{lemma}
\begin{proof}
Homogeneity property asserts that coordinate transformation between an IF and a reference frame which moves with uniform velocity with respect to it, is linear. Since linear transformation preserves rectilinear motion, therefore a reference frame which moves with uniform velocity with respect to an IF, is also an IF. If we start with an inertial frame $F$, and do not declare $F'$ an IF apriori, homogeneity assures $F'$ to be an IF where Newton's first law also holds. According to relativity principle $F'$ is also an IF. Therefore, homogeneity property of space and time is compatible with principle of relativity. 
\end{proof}
Relativity principle alone implies the spacetime coordinate transformation between two IFs is necessarily projective function (ratio of two linear functions); linear transformation is a special case, one can however restrict themselves to linear transformations for physical reasons \cite{rindler}. 
\begin{theorem}\label{thm1}
$F$ moves with uniform speed $v$ relative to $F'$ along the negative $X'$ axis. \end{theorem}
\begin{proof}
From the earlier discussion, it is quite evident that an observer in $F'$ also sees that reference frame $F$ is moving in a translation motion along $X'$ axis. Not only the motion of $F$ with respect to $F'$ has to be a translation motion but also it is a motion with uniform velocity because otherwise the inverse relation of the Eq. \eqref{x1}-Eq. \eqref{t1} would not be linear which is Mathematically impossible. This argument can also be made in an alternative way. We start from the definition of IF. If a reference frame is inertial, any other reference frame which moves with uniform velocity is an IF. If $F$ moves in a nonuniform translation motion (motion with some acceleration) with respect to $F'$, then it becomes impossible for both of the frame to be inertial. This puts reference frame $F$ with a very special frame such that observer in $F$ concludes that $F'$ moves with uniform motion, hence $F'$ is an IF relative to $F$ but the observers in $F'$ see $F$ as a noninertial frame. Since, all IFs must have equal status, therefore both $F$ and $F'$ are IFs. With this in mind and from the discussions about the transformation equations \eqref{x1}-\eqref{t1}, we now conclude that $F$ moves with some uniform speed $v'$ along the negative $X'$ axis of $F'$. For $v'\neq v$, $F$ relative to $F'$ moves either slower or faster than $F'$ relative to $F$. This disagreement of speed between the observers in $F$ and $F'$ implies that either between $F$ and $F'$, there is a special preferred IF, or positive $X'$ and negative $X'$ directions are not in same footing. Since there is no preference between inertial frames, and space is isotropic, hence $v'=v$. 
\end{proof}
This is velocity reciprocity principle, a crucial step in the derivation, here we present a proof by contradiction $3+1$ spacetime dimension. Einstein in his derivation proved this principle. As said already, Einstein's derivation begins with the observation of mismatch of simultaneity between two IFs, that provides the time transformation equation between two IFs. However, in our case, we do not have a very specific form of time transformation equation \eqref{t1} yet. We arrive at the velocity reciprocity principle without using Lorentz invariant speed of light. In some literature, this principle is assumed to begin with \cite{moller1972theory,lee,BC,doi:10.1119/10.0010234}. On a note, we would like to suggest the readers to have a look at the following references which discuss the necessity of isotropy property of space for the validity of velocity reciprocity principle:-Ref. \cite{gorini,1970JMP....11.2226G, 1971CMaPh..21..150G, ga,doi:10.1119/10.0009219,schroder1990special}. 
\begin{corollary}\label{cor1}
If we construct another reference frame $F ''$ moving with uniform velocity with respect to $F'$, $F''$ is an IF moving with uniform velocity with respect to $F$.
\end{corollary}
\begin{proof}
By definition, $F''$ is an inertial frame, in proving the Theorem \ref{thm1}, we have proved already that $F$ and $F'$ moves with uniform velocity in opposite directions relative to $F'$ and $F$ respectively. Since $F''$ is also an IF therefore it is either at rest or is moving with uniform velocity with respect to $F$.  
\end{proof}
\begin{corollary}\label{cor2}
$\gamma (v)=\gamma (-v)$, $\phi(v)=1$, $c(v)=d(v)=0$, $a(v)=\gamma (v)$ and $b(v)=\frac{\gamma (v)}{v}\left(\frac{1}{\gamma (v)^2}-1\right)$.
\end{corollary}
\begin{proof}
From Theorem \ref{thm1}, we have the inverse transformation relation of the Eq. \eqref{x1}-Eq. \eqref{t1}:
\begin{eqnarray}
& x=\gamma (-v)(x'+vt') \label{x1'}\\
& y=\phi (-v) y', \label{y1'}\\
& z=\phi (-v) z', \label{z1'}\\
& t=a(-v)t' +b(-v)x'+c(-v)y'+d(-v)z'. \label{t1'}
\end{eqnarray}

Moving rod's length change (contraction or elongation) from its rest value as viewed from the stationary IF, has to be independent of its direction of motion because of isotropy. Therefore $\gamma (v)=\gamma (-v)>0$, and $\phi (v)=\phi (-v)>0$.  Eq. \eqref{y1}, and Eq. \eqref{y1'} provide $\phi (v)\phi (-v)=1$; since $\phi(v)$ is an even function of $v$, therefore $\phi(v) ^2 =1$ $\Rightarrow \phi (v)=1\because \phi (v)>0$ by lemma \ref{lemma 1}. By using $\gamma (v)=\gamma (-v)$, and inserting the right hand side expression of $x'$ of Eq. \eqref{x1} into the Eq. \eqref{x1'}, we find the linear transformation relation of $t'$ in terms of $t$ and $x$. Hence we are finally left with a general form of rotation free LT as below:- 
\begin{eqnarray}
& x'=\gamma (v)(x-vt), \qquad x=\gamma (v)(x'+vt'); \label{x2}\\
& y'= y, \qquad y= y';\label{y2}\\
& z'= z, \qquad z= z';\label{z2}\\
& t'=\gamma (v)\left(t+\frac{x}{v}\left(\frac{1}{\gamma (v)^2}-1\right)\right), \qquad t=\gamma (v)\left(t'-\frac{x'}{v}\left(\frac{1}{\gamma (v)^2}-1\right)\right). \label{t2}
\end{eqnarray}
\end{proof}
We prove $c(v)=d(v)=0$ by using the velocity reciprocity principle, i.e., by using Theorem \ref{thm1}. However, one can prove it in a different way, i.e., by using a symmetry argument that $t'$ in the Eq. \eqref{t1'} should be independent of $y$ and $z$ because any position at $F$ coordinate perpendicular to motion of $F$ with respect to $F'$ must have equal status. This is what stated in the Ref. \cite{schroder1990special, bergmann1976introduction}. Some of the authors have preferred to work in the $1+1$D spacetime \cite{sen,peli,gorini,lee,doi:10.1119/10.0009219, doi:10.1119/10.0010234}. In this regard, we suggest the readers for a further read on Gorini and their collaborators' work which conclude that it is only Galilean transformation and Lorentz transformation which are compatible with the isotropy property of space in any spacetime dimension \cite{1970JMP....11.2226G, 1971CMaPh..21..150G}.

$\gamma (v)=1$ is certainly Galilean transformation.
We already know that $\gamma(v)$ is an even function of $v$, we can Taylor expand it in a general form around $v=0$, provided $\gamma (v=0)=1$:-
\begin{equation}\label{tform}
\gamma (v)=1+B_2 \frac{v^2}{u^2} + B_4 \frac{v^4}{u^4} + B_6\frac{v^6}{u^6}+..=\sum _{n=0}^{n=\infty} B_{2n}\frac{v^{2n}}{u^{2n}}.
\end{equation}
$B_{2n}$s are dimensionless quantities, $B_0=1$, and constant $u$ has dimension of speed. From such a Taylor expansion, one might guess the existence of a universal speed scale $u=c$, which remains same relative to all inertial frames. Therefore propagation of such a signal satisfies: $c^2dt^2-dx^2-dy^2-dz^2=c^2dt'^2-dx'^2-dy'^2-dz'^2=0$; by using Eq. \eqref{x2} to Eq. \eqref{t2}, one obtains from the first equality, $\gamma (v)=\frac{1}{\sqrt{1-\frac{v^2}{c^2}}}$. However, the Taylor expansion \eqref{tform} doesn't guarantee that such a speed would exist, $u$ may as well be a constant number having dimension of speed with suitable $B_{2n}$s. At this level, it is purely intuitive. 

\noindent
{\it Velocity addition:-} Let us consider an IF moving with uniform velocity $v'$ along positive $X'$ axis, relative to $F'$. The velocity of $F''$ with respect to $F$ can be found from the Eq. \eqref{x2}-\eqref{t2}, it is  $v''=\frac{v+v'}{1-\frac{v'}{v}\left(\frac{1}{\gamma (v)^2} -1\right)}$ along the positive $X$ axis of $F$. Evidently, finding a velocity of the IF $F''$ with respect to the IF $F$ is consistent with the corollary \ref{cor1}.
\begin{corollary}
$\gamma (v)=\frac{1}{\sqrt{1+k v^2}}$, where constant $k\in \mathbb{R}$.
\end{corollary}
\begin{proof}
The coordinate transformation equations between $F''$ and $F'$ take the form of Eq. \eqref{x2}-Eq. \eqref{t2} with $v$ replaced by $v'$. Therefore, using the Eq. \eqref{x2}-\eqref{t2}, we have
\begin{eqnarray}
& x''=\gamma (v')\gamma (v)\left[1-\frac{v'}{v}\left(\frac{1}{\gamma (v) ^2}-1\right)\right]x -\gamma (v)\gamma (v')(v+v')t.  \label{x3} \\
& y''=y'=y, \label{y3}\\
& z''=z'=z, \label{z3}\\
& t''= \gamma (v')\gamma (v)\left[1-\frac{v}{v'}\left(\frac{1}{\gamma (v') ^2}-1\right)\right]t+\gamma (v')\gamma (v)\left[\frac{1}{v}\left\{\frac{1}{\gamma (v) ^2}-1\right\}+\frac{1}{v'}\left\{\frac{1}{\gamma (v') ^2}-1\right\}\right]x
\label{t3}
\end{eqnarray}
The coordinate transformation equations between $F''$ and $F$ take the form of Eq. \eqref{x2}-Eq. \eqref{t2} with $v$ replaced by $v''$. Therefore, in the RHS of the first term in the Eq. \eqref{x3}, and the Eq. \eqref{t3}, we have 
\begin{equation}
 \gamma (v'') =\gamma (v')\gamma (v)\left[1-\frac{v'}{v}\left(\frac{1}{\gamma (v) ^2}-1\right)\right]=  \gamma (v')\gamma (v)\left[1-\frac{v}{v'}\left(\frac{1}{\gamma (v') ^2}-1\right)\right].
\end{equation}
Comparing the second equality, we find
\begin{equation}\label{hehe}
\frac{\frac{1}{\gamma (v')^2}-1}{\frac{1}{\gamma (v)^2}-1}=\frac{v'^{2}}{v^2}.
\end{equation}
$\gamma (v)$ is an even function of the form \eqref{tform}, hence $\frac{1}{\gamma (v)^2}$ is also an even function. 
Therefore, according to the Eq. \eqref{hehe}, we have
\begin{equation}
\gamma (v)^2=\frac{1}{1+k v^2}.   
\end{equation}
Since by lemma \ref{lemma 1}, $\gamma (v) >0$, therefore
\begin{equation}
\gamma (v)=\frac{1}{\sqrt{1+k v^2}}.   
\end{equation}
\end{proof}
This is how the group structure of the transformation leaves us with a specific form of the transformation \cite{gorini,peli}.
\begin{corollary}
$k=0$ is Galilean transformation, $k<0$ is Lorentz transformation, $k>0$ is not possible.
\end{corollary}
\begin{proof}
$k=0$ is the trivial case, Galilean transformation, i.e., $\gamma (v)=1$.

\noindent 
$\gamma (v)$ is a dimensionless quantity, $k$ has dimension of inverse of velocity square, i.e., $k=\pm \frac{1}{c^2}$, $c$ is a real constant number having dimension of speed.

\noindent
$k<0$, i.e., $\gamma (v)\geq 1$ implies Lorentz transformation.$\gamma(v)=\frac{1}{\sqrt{1-\frac{v^2}{c^2}}}$. From the previous discussion, we note $c$ is a speed which is invariant under Lorentz transformation. Since $\gamma (v)$ has to be a real number for a physical coordinate transformation, $c$ is the maximum allowable speed. There exists only a single speed $c$ which is Lorentz invariant.

\noindent
When we have $k=\frac{1}{c^2}$, i.e., $\gamma (v)\leq 1$, there is no maximum speed bound. From the relation \eqref{t3}, we have $t=\frac{1}{\sqrt{1+\frac{v^2}{c^2}}}\left( t'-\frac{v x'}{c^2}\right)$. Let us imagine a particle moving in $F'$ along the positive $X'$ axis with speed $u$. Therefore, for $u>\frac{c^2}{v}$, the particle time travels in the past in the reference frame $F$. If we allow time travel in the past like what happens in science fictions, $k>0$ is possible. In reality, $k>0$ would allow possibilities which simply violate causality. For $k<0$, it is straight forward to find that due to the maximum speed bound, time travel in the past isn't possible.

\end{proof}
Nevertheless, we have hypothetical particle which moves faster than light violating causality \cite{1962AmJPh..30..718B,PhysRev.162.1274, LIBERATI2002167}.
\section{Conclusion}
Therefore, principle of relativity along with homogeneity of space and time and isotropy of space naturally gives rise to two possibilities, either we have Galilean transformation where space and time is absolute, or we have Lorentz Transformation with a boost invariant speed $c$ as the upper limit of all speeds. Nature happens to choose the second possibility, the more interesting one in our opinion, and $c$ happens to be the speed of Electromagnetic wave in vacuum. 
We end our discussion with a quote by Einstein ``The most incomprehensible thing about the world is that it is comprehensible."
\section{Acknowledgments}
The author thanks Dr. Sang-Shin Baak for useful discussion.
This work has been supported by the National Research Foundation of Korea under Grants No.~2017R1A2A2A05001422 and No.~2020R1A2C2008103.

\bibliography{arl}
\end{document}